\documentclass[runningheads]{llncs}

\usepackage{graphicx}
\usepackage{amsmath}
\usepackage{amssymb}
\usepackage[labelfont=bf,size=small]{caption}
\usepackage[labelfont=normalfont]{subfig}
\usepackage{verbatim}
\usepackage{enumerate}
\usepackage{times}
\usepackage{tabularx}
\usepackage{todonotes}
\usepackage{hyperref}
\usepackage{wrapfig}
\usepackage{timetravel}
\usepackage{color}

\setCounters{theorem}

\title{Smooth Orthogonal Drawings of Planar Graphs}

\author{
  Md.~Jawaherul~Alam\inst{1}
    \and
  Michael~A.~Bekos\inst{2}
    \and
  Michael~Kaufmann\inst{2}
    \and
  Philipp~Kindermann\inst{3}
    \and
  Stephen~G.~Kobourov\inst{1}
    \and
  Alexander~Wolff\inst{3}}

\institute{%
    Department of Computer Science, University of Arizona, USA\\
    \email{\{mjalam,kobourov\}@cs.arizona.edu}
    \and
    Wilhelm-Schickhard-Institut f\"ur Informatik, Universit\"at T\"ubingen, Germany\\
    \email{\{bekos,mk\}@informatik.uni-tuebingen.de}
    \and
    Lehrstuhl f\"ur Informatik I, Universit\"at W\"urzburg, Germany\\
    \email{http://www1.informatik.uni-wuerzburg.de/en/staff}
}

\authorrunning{Md.~Alam et al.}
\titlerunning{Smooth Orthogonal Drawings of Planar Graphs}

\newenvironment{pf}{\begin{proof}}{\qed\end{proof}}
\newcommand{\SC}[1]{SC$_{#1}$}
\newcommand{\OC}[1]{OC$_{#1}$}

\newcommand{\eps}{\ensuremath{\varepsilon}}
\graphicspath{{figures/},{pic/}}

\renewcommand{\textcolor}[2]{#2}

\begin{document}

\renewcommand{\topfraction}{0.99}
\renewcommand{\bottomfraction}{0.99}
\renewcommand{\textfraction}{0.01}

\maketitle

\begin{abstract}
  In \emph{smooth orthogonal layouts} of planar graphs, every edge
  is an alternating sequence of axis-aligned segments and circular arcs with common
  axis-aligned tangents.  In this paper, we study the problem of
  finding smooth orthogonal layouts of low \emph{edge complexity}, that is,
  with few segments per edge.  We say that a graph
  has \emph{smooth complexity} $k$---for short, an \SC{k}-layout---if it
  admits a smooth orthogonal drawing of edge complexity at most~$k$.

  Our main result is that every 4-planar graph has an \SC2-layout.
  While our drawings may have super-polynomial area, we show that, for
  3-planar graphs, cubic area suffices.  Further, we show that every
  biconnected 4-outerplane graph admits an \SC1-layout.  On the negative
  side, we demonstrate an infinite family of biconnected 4-planar graphs that
  requires exponential area for an \SC1-layout.  Finally, we present an
  infinite family of biconnected 4-planar graphs that does
  not admit an \SC1-layout.
\end{abstract}

\section{Introduction}
\label{sec:intro}

In the visualization of technical networks such as the structure of VLSI chips
\cite{lei80} 
or UML diagrams~\cite{s-esadu-GD97} there is a strong tendency to draw
edges as rectilinear paths.
The problem of laying out networks in such a way is called
\emph{orthogonal graph drawing} and has been studied extensively.  For
drawings of (planar) graphs to be readable, special care is
needed to keep the number of bends small.  In a seminal work,
Tamassia~\cite{t-eggmn-87} showed that one can efficiently minimize
the total number of bends in orthogonal layouts of \emph{embedded
  4-planar graphs}, that is, planar graphs of maximum degree~4 whose
combinatorial embedding (the cyclic order of the edges around each vertex)
is given.  In contrast to this, minimizing the number of bends over
all embeddings of a 4-planar graph is NP-hard~\cite{gt-ccurp-01}.

In a so far unrelated line of research, circular-arc drawings of
graphs have become a popular matter of research in the last few years.
Inspired by American artist Mark Lombardi
(1951--2000), 
 Duncan et al.~\cite{degkn-ldg-JGAA12} introduced and studied \emph{Lombardi drawings}, which are circular-arc
drawings with the additional requirement of \emph{perfect angular
  resolution}, that is, for each vertex, all pairs of consecutive
edges form the same angle.  Among others, Duncan et al.\ treat
drawings of $d$-regular graphs
where all vertices have to lie on one circle.  They show that under
this restriction, for some subclasses, Lombardi drawings can be constructed
efficiently, whereas for the others, the problem is NP-hard.  They also
show~\cite{degkn-dtpar-DCG13} that trees can always be Lombardi drawn
in polynomial area, whereas straight-line drawings with perfect
resolution may need exponential area.

Very recently, Bekos et al.~\cite{BKKS12} introduced the
\emph{smooth orthogonal} graph layout problem that combines the two worlds;
the rigidity and clarity of orthogonal layouts with the artistic
style and aesthetic appeal of Lombardi drawings.  Formally, a smooth
orthogonal drawing of a graph is a drawing on the plane where
(i)~each vertex is drawn as a point; (ii)~edges leave and enter
vertices horizontally or vertically, (iii)~each edge is drawn as an alternating
sequence of axis-aligned line segments and circular-arc segments
such that consecutive segments have a common \emph{horizontal} or
\emph{vertical} tangent at their intersection point.  In the case of
(4-) planar graphs, it is additionally required that (iv)~there are
no edge-crossings.  Note that, by construction, (smooth) orthogonal
drawings of 4-planar graphs have angular resolution within a factor
of two of optimal.

\begin{wrapfigure}[13]{r}{.35\textwidth}
  \vspace{-5ex} 
  \centering
  \includegraphics[width=\linewidth]{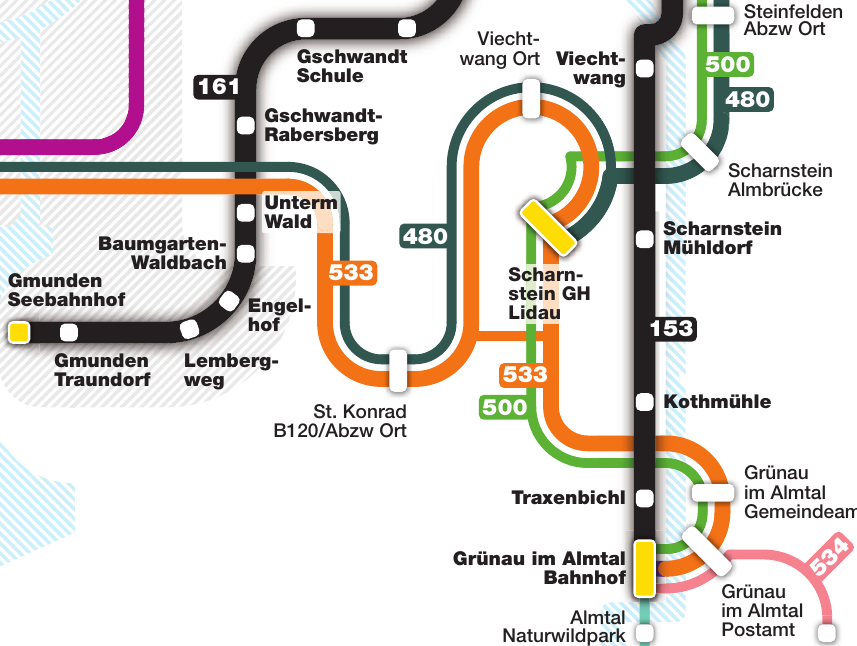}
  \caption{Clipping of the public transport map
    \emph{Gmunden -- V\"ocklabruck -- Salzkammergut}, Austria~\cite{url-map}}
  \label{fig:transport_map}
\end{wrapfigure}
Figure~\ref{fig:transport_map} shows a real-world example: a smooth
orthogonal drawing of an Austrian regional bus and train map.
Extending our model, the map has (multi-) edges that enter vertices
diagonally (as in \emph{Gr\"unau im Almtal Postamt}; bottom right).

For usability, it is important to keep the visual complexity of such
drawings low.  In a (smooth) orthogonal drawing, the
\emph{complexity} of an edge is the number of segments it consists
of, that is, the number of inflection points plus one.  Then, a
natural optimization goal is to minimize, for a given (embedded)
planar graph, the \emph{edge complexity} of a drawing, which is
defined as the maximum complexity over all edges.  We say that a
graph has \emph{orthogonal complexity $k$} if it admits an
orthogonal drawing of edge complexity at most $k$, for short, an
\emph{\OC{k}-layout}. Accordingly, we say that a graph has
\emph{smooth complexity~$k$} if it admits a smooth orthogonal
drawing of edge complexity at most~$k$,
for short, an \emph{\SC{k}-layout}.  
We seek for drawings of 4-planar graphs with low smooth complexity.

\paragraph{Our Contribution.}

Known results and our contributions to smooth orthogonal
drawings are shown in Table~\ref{table:res}.
The main result of our paper is that any 4-planar graph admits an
\SC{2}-layout.  We start with the biconnected case (see
Section~\ref{sec:bicon}) and then turn towards general 4-planar
graphs (see Section~\ref{sec:arbit}).  Our upper bound of~2 for the
smooth complexity of 4-planar graphs improves the previously known
bound of~3 and matches the corresponding lower bound~\cite{BKKS12}. In contrast to
the known algorithm for \SC3-layout \cite{BKKS12}, which is based on an algorithm
for \OC3-layout of Biedl and Kant~\cite{bk-bhogd-CGTA98}, we use an
algorithm of Liu at al.~\cite{lms-la2be-DAM98} for \OC3-layout,
which avoids so-called S-shaped edges (see Figure~\ref{fig:shapes}b,
top). Such edges are not desirable since they impose strong
restrictions on vertex positions in a smooth orthogonal layout (see
Figure~\ref{fig:shapes}b, bottom). \textcolor{red}{Our construction requires more
than polynomial area. Therefore, we made no effort in proving
a bound.}

Further, we
prove that every biconnected 4-outerplane graph admits an \SC1-layout
(see Section~\ref{sec:smooth-1}), expanding the class of graphs with
\SC{1}-layout from triconnected or Hamiltonian 3-planar
graphs~\cite{BKKS12}.  Note that in our result, the outerplane
embedding can be prescribed, while in the other results the algorithms
need the freedom to choose an appropriate embedding.

We complement our positive results by the following two negative results;
see the appendix.

\wormhole{exponential-area}
\begin{theorem}\label{thm:exp}
  There is an infinite family of graphs that require exponential
  area if they are drawn with \SC1.
\end{theorem}

%
So far, such a family
of graphs has only been known under the additional, rather strong
restriction of a fixed \emph{port assignment} \cite[Thm.~5,
Fig.~7]{BKKS12}.  A port assignment prescribes, for each edge, in which
direction it must enter its endpoints.

\wormhole{OCtwo-not-SCone}
\begin{theorem}
  \label{th:sc2}
  There is an infinite family of biconnected 4-planar graphs that
  admit \OC2-layouts but do not admit \SC{1}-layouts.
\end{theorem}

%
So far, the only graphs known not to admit an \SC1-layout were the
octahedron (which is the only 4-planar graph that needs~\OC4) and a family of
graphs 
with a fixed triangular outer face.

\begin{table}[tb]
    \caption{Comparison of our results to the results of Bekos et
      al.~\cite{BKKS12}.}
    \label{table:res}

    \centering
    \begin{tabular}{|l|c|c|c|c|}
        \hline
        graph class & \multicolumn{3}{c|}{our contribution} &
        Bekos et al.~\cite{BKKS12} \\
        \cline{2-4}
        &complexity & area & reference & \\
        \hline
        biconnected 4-planar & \SC2 & super-poly & Theorem~\ref{thm:bicon-sc2} & \SC3 \\
        4-planar & \SC2 & super-poly & Theorem~\ref{thm:ar-smooth-2}&\\
        3-planar & \SC2 & $\lfloor {n^2}/4 \rfloor \times \lfloor n/2 \rfloor$ & Theorem~\ref{thm:bi-poly-3} &\\
        biconnected 4-outerplane & \SC1 & exponential & Theorem~\ref{thm:biconnected-outerplanar} &\\
        triconnected 3-planar & & & & \SC1 \\
        Hamiltonian 3-planar & & & & \SC1 \\
        poly-area & $\not\supseteq$\SC1 & & Theorem~\ref{thm:exp} & \\
        \OC3, octahedron & & & & $\not\subseteq$\SC1 \\
        \OC2 & $\not\subseteq$\SC1 & & Theorem~\ref{th:sc2} & \\
        \hline
    \end{tabular}
\end{table}

\section{Smooth Layouts for Biconnected 4-Planar Graphs}
\label{sec:bicon}

In this section, we prove that any biconnected 4-planar graph admits
an \SC2-layout. Given a biconnected 4-planar graph, we first
compute an \OC3-layout, using an algorithm of Liu et
al.~\cite{lms-la2be-DAM98}. Then we turn the result of their
algorithm into an \SC2-layout.

\textcolor{red}{Liu et al.\ choose two vertices~$s$ and~$t$ and compute an
$st$-ordering of the input graph. An $st$-ordering is an
ordering $(s=1,2,\dots,n=t)$ of the vertices such that
every~$j$ ($2<j<n-1$) has neighbors $i$ and $k$ with $i<j<k$.}
Then they draw vertices~1 and~2 on a horizontal grid line, \emph{row}~1, connecting
them by a so-called U-shape; see Fig.~\ref{fig:shapes}f.  They go
through the other vertices as prescribed by the $st$-ordering, placing
vertex~$i$ in row~$i-1$.  Calling an edge of which exactly one
end-vertex is already drawn an \emph{open edge}, they maintain the
following invariant:
\begin{enumerate}[$({I}_1)$]
\item In each iteration, every open edge is associated with a \emph{column}
  (a vertical grid line).
\end{enumerate}
An algorithm of Biedl and Kant~\cite{bk-bhogd-CGTA98} yields an
\OC3-layout similar to that of Liu et al.  However, Liu et al.
additionally show how to modify their algorithm such that it produces
\OC3-layouts without S-shapes; see Fig.~\ref{fig:shapes}b (top).  In an \SC2-layout, S-shapes
are composed of two quarter-circles; see Fig.~\ref{fig:shapes}b
(bottom); they are undesirable as they force their endpoints to lie on
a line of slope $\pm 1$.

In their modified algorithm, Liu et al.\ search for paths in the
drawing that consist only of S-shapes; every vertex lies on at most
one such path.  They place all vertices on such a path in the same
row, without changing their column.  This essentially converts all
S-shapes into horizontal edges.  Now every edge (except $(1,2)$ and
$(1,n)$) is drawn as a vertical segment, horizontal segment, L-shape,
or C-shape; see Fig.~\ref{fig:shapes}.  The edge $(1,2)$ is drawn as a
U-shape and the edge $(1,n)$, if it exists, is either drawn as a
C-shape or (only in the case of the octahedron) as a three-bend edge
that uses the left port of vertex~1 and the top port of vertex~$n$.

\begin{figure}[tb]
  \includegraphics{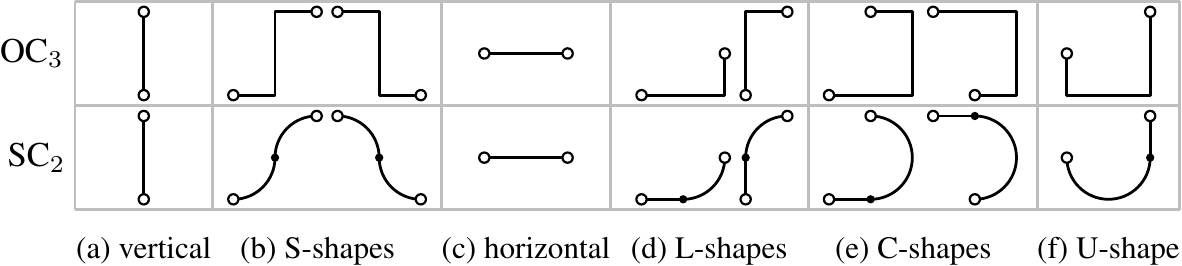}
  \caption{Converting shapes from the \OC3-layout to \SC2.}
  \label{fig:shapes}
\end{figure}

We convert the output of the algorithm of Liu et al.\ from \OC3
to~\SC2.  The coordinates of the vertices and the port assignment of
their drawing define a (non-planar) \SC2-layout using the conversion
table in Fig~\ref{fig:cases}.  In order to avoid crossings, we
carefully determine new vertex positions scanning the drawing of Liu
et al.\ from bottom to top.

\begin{figure}[tb]
  \centering
  \includegraphics[scale=.72]{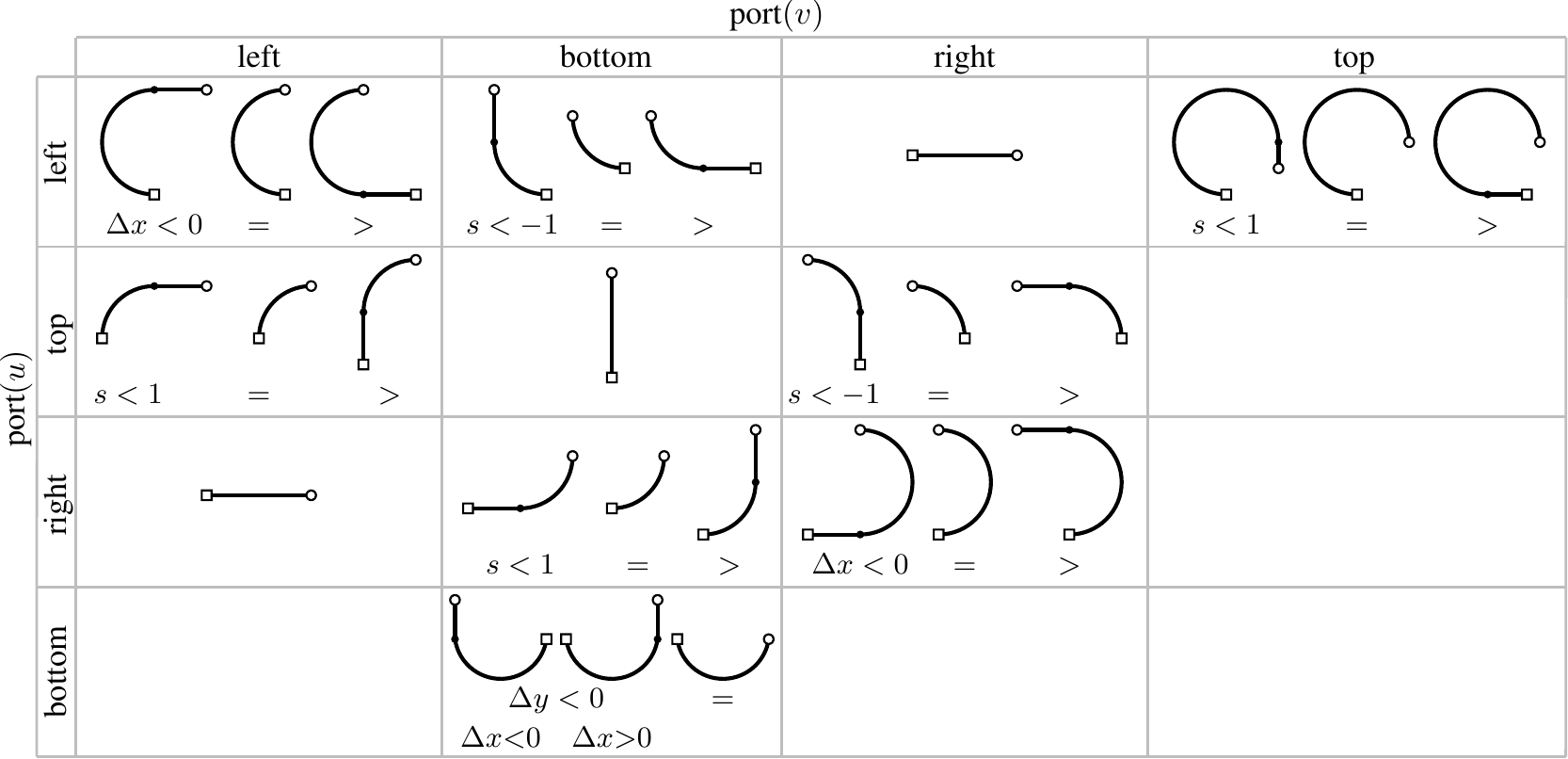}
  \caption{Cases for drawing the edge $(u,v)$ based on the port
    assignment.  In each case, $u$ is the lower of the two vertices
    ($y(u)<y(v)$).  As shorthand, we use $\Delta x=x(u)-x(v)$,
    $\Delta y=y(u)-y(v)$, and $s=\text{slope}(u,v)=\Delta x/\Delta y$.}
  \label{fig:cases}
\end{figure}

We now introduce our main tool for the conversion:
a \emph{cut}, for us, is a $y$-monotone curve consisting of
horizontal, vertical, and circular segments that divides the current
drawing into a left and a right part, and only intersects horizontal
segments and semi-circles of the drawing. In the following, we describe how one can find such a cut
from any starting point at the top of the drawing; see
Fig.~\ref{fig:cuts}.
(In spite of the fact that we define the cut going from top to bottom,
``to its left'' will, as usually, mean ``with smaller $x$-coordinate''.)

When such a cut encounters a vertex $u$ to its right with an outgoing
edge associated with its left port, then the cut continues by
passing through the segment incident to $u$. On the other hand, if
the port has an incoming L-shaped or C-shaped edge, the
cut just follows the edge.  The case when the cut encounters a
vertex to its left is handled symmetrically.  

\textcolor{red}{Let~$v$ be a vertex incident to two incoming C-shapes~$(u,v)$ and~$(w,v)$.
If $y(w)\le y(u)$ we call the C-shape~$(u,v)$ \emph{protected} by~$(w,v)$; otherwise, we
call it \emph{unprotected}.} In order to ensure
that a cut passes only through horizontal segments and that our
final drawing is planar, our algorithm will maintain the
following new invariants:
\begin{enumerate}[$({I}_1)$]
\setcounter{enumi}{1}
\item \textcolor{red}{An L-shape never contains
  a vertical segment (as in Fig.~\ref{fig:shapes}d right);
  it always contains a horizontal
  segment (as in Fig.~\ref{fig:shapes}d left) or no
  straight-line segment.}
\item \textcolor{red}{An unprotected C-shape never contains
  a horizontal segment incident to its top vertex
  (as in Fig.~\ref{fig:shapes}e right);
  it always contains a horizontal segment incident to its bottom
  vertex(as in Fig.~\ref{fig:shapes}e left) or no
  straight-line segment.}
\item The subgraph induced by the vertices that have already been
  drawn has the same embedding as in the drawing of Liu et al.
\end{enumerate}
\textcolor{red}{Below, we treat L- and C-shapes of complexity~1 as if they had
a horizontal segment of length~0 incident to their bottom vertex. Note that 
every cut moves around the protected C-shapes, so it will never intersect
their semi-circular segments.}
Now we we are ready to state the main Theorem of this Section by
presenting our algorithm for \SC2-layouts.


\begin{theorem}\label{thm:bicon-sc2}
  Every biconnected 4-planar graph admits an \SC{2}-layout.
\end{theorem}

\begin{pf}
  In the drawing~$\Gamma$ of Liu et al., vertices are arranged in rows.
  Let $V_1,\dots,V_r$ be the partition of the vertex set~$V$ in rows
  $1,\dots,r$.  Following Liu et al., the vertices in each such set
  induce a path in~$G$.
  We place vertices in the order $V_1,\dots,V_r$.  In this process, we
  maintain a planar drawing~$\Gamma'$ and the invariants~$I_1$
  to~$I_4$.  As Liu et al., we place the vertices on the integer grid.
  We deal with the special edges $(1,2)$ and $(1,n)$ at the
  end, leaving their ports, that is, the bottom and left port of
  vertex~$1$ and the top port of vertex~$n$, open.

\begin{figure}[tb]
  \begin{tabular}{cccc}
    \includegraphics[page=1]{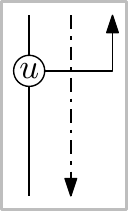} &
    \includegraphics[page=2]{figures/cutting-line} &
    \includegraphics[page=3]{figures/cutting-line} &
    \includegraphics{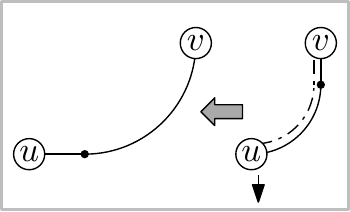} \\
    (a) open edge & (b) L-shape & (c) C-shape &
    (d) Maintaining $I_2$\\
  \end{tabular}
  \caption{Finding a cut.}
  \label{fig:cuts}
\end{figure}

  For invariant~$I_1$, we associate each open edge with the column on
  which the algorithm of Liu et al.\ places it.  If their algorithm
  draws the first segment of the open edge horizontally (from the
  source vertex to the column), we use the same segment for our
  drawing. We use the same ports for the edges as their algorithm.
  Thus, our drawing keeps the embedding of Liu et al., maintaining
  invariant~$I_4$.

  Assume that we have placed $V_1,\dots,V_{i-1}$ and that the vertices
  in~$V_i$ are $v_1,\dots,v_c$ in left-to-right order
  (the case~$v_1=v_c$ is possible). Vertex~$v_j$
  ($1\le j \le c$) is placed in the column with which the edge
  entering the bottom port of~$v_j$ is associated.  If the left port
  of~$v_1$ is used by an incoming (L- or C-shaped) edge~$e=(u_1,v_1)$, we
  place~$v_1$ (and the other vertices in $V_i$) on a row high enough
  so that a smooth drawing of~$e$ does not create any crossings with
  edges lying on the right side of~$e$ in~$\Gamma$;
  see Fig.~\ref{fig:handling_c-shapes}b.
  
  \begin{figure}[tb]
    \begin{tabular}{cccccc}
      \includegraphics[page=1]{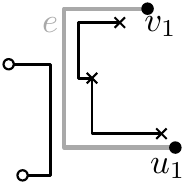} &
      \includegraphics[page=2]{pic/handling_c-shape} &
      \includegraphics[page=3]{pic/handling_c-shape} &
      \includegraphics[page=4]{pic/handling_c-shape} &
      \includegraphics[page=1]{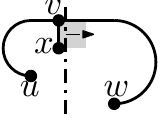} &
      \includegraphics[page=2]{pic/c-shape_special} \\
      (a) \OC3-layout & (b) moving $v_1$ & (c) finding a cut &
      (d) \SC2-layout & \multicolumn{2}{c}{(e)(f) protected C-shape}\\
    \end{tabular}
    \caption{Handling C-shapes}
    \label{fig:handling_c-shapes}
  \end{figure}


  In order to make sure that the new drawing of~$e$ does not create
  crossings with edges on the left side of~$e$ in~$\Gamma$, we need to
  ``push'' those edges to the left of~$e$.  We do this by computing a
  cut that starts from $v_1$, separates the vertices and edges that lie
  on the left side of~$e$ in~$\Gamma$ from those on the right side,
  passes~$u_1$ slightly to the left, and continues downwards as
  described above; see Fig.~\ref{fig:handling_c-shapes}c. Since, by
  invariant~$I_4$, our drawing so far is planar and each edge is drawn
  in a $y$-monotone fashion, we can find a cut, too, that is
  $y$-monotone.  \textcolor{red}{We move everything on the left side
  of the cut further left such that~$e$ has no more crossings.
  Note that the cut intersects only horizontal edge segments.
  These will simply become longer by the move.}

  \textcolor{red}{Let $\Delta x_i = x(v_i)-x(u_i)$ and $\Delta y_i = y(v_i)-y(u_i)$
  for $i=1,\dots,c$. It is possible that the drawing of~$e$ violates
  invariant~$I_3$---if $u_1$ lies to the left of~$v_1$.  We consider two
  cases. First, assume that the edge~$(u_1,v_1)$ is the only incoming
  C-shape at $v_1$. In this case, we simply define a cut that starts
  slightly to the right of~$v_1$, follows~$e$, intersects~$e$ slightly
  to the left of~$u_1$, and continues downwards.  Then we move everything
  on the left side of the cut by $\Delta x_1+1$ units to the left.
  Next, assume that there is another C-shape~$(w_1,v_1)$ entering the
  right port of~$v_1$; see Fig.~\ref{fig:handling_c-shapes}e.
  We assume w.l.o.g. that $y(w_1)\le y(u_1)$. Let~$(x_1,v_1)$ be the
  edge incident to the bottom vertex of $v_1$. In this case, we 
  first find a cut that starts slightly to the right of~$v_1$,
  follows~$(x_1,v_1)$, passes~$x_1$ slightly to the right,
  and continues downwards. Then we move everything
  on the right side of the cut by $y(v_1)-y(x_1)$ units to the right.
  Thus, there is an empty square to the right of~$x_1$ with size~$y(v_1)-y(x_1)$-
  Now we place~$v_1$ at the intersection of the diagonal through~$x_1$
  with slope~1 and the vertical line through~$w_1$. Because of this
  placement, we can draw the edge~$(x_1,v_1)$ by using to quarter-circles
  with a common horizontal tangent in the top right corner of the empty square;
  see Fig.~\ref{fig:handling_c-shapes}f. Note that the edge~$(u_1,v_1)$ is
  protected by~$(w_1,v_1)$, so it can have a horizontal segment incident to~$v_1$.  
  This establishes~$I_3$.} 

  It is also possible that the drawing of~$e$ violates
  invariant~$I_2$---if~$\text{slope}(u_1,v_1)>1$. In this case we
  define a cut that starts slightly to the left of~$v_1$,
  intersects~$e$ and continues downwards. Then we move everything on
  the left side of the cut by $\Delta y_1$ units
  to the left.  This establishes~$I_2$.

  We treat~$v_c$, the rightmost vertex in the current row,
  symmetrically to~$v_1$.

  Now we have to treat the edges entering the vertices~$v_1,\dots,v_c$
  from below. Note that these edges can only be vertical or
  L-shaped. Vertical edges can be drawn without violating the
  invariants. However, invariant~$I_2$ may be violated if an
  edge~$e_i=(u_i,v_i)$ entering the bottom port of vertex~$v_i$
  is L-shaped; see Fig.~\ref{fig:cuts}d.
  Assume that $x(u_i)<x(v_i)$. In this case we
  find a cut that starts slightly to the left of $v_i$,
  follows $e_i$, intersects~$e_i$ slightly to the right of $u_i$,
  and continues downwards.  Then we move everything on the left side
  of the cut by~$\Delta y_i$ units to the left.
  This establishes~$I_2$.  We handle the case $x(u_i)>x(v_i)$
  symmetrically.

  We thus place the vertices row by row and draw the incoming edges
  for the newly placed vertices, copying the embedding of the
  current subgraph from~$\Gamma$.  This completes the drawing of
  $G-\{(1,2),(1,n)\}$.  Note that vertex $1$ has no incoming
  edge and vertex $2$ has only one incoming edge, that is, $(1,2)$.
  Thus, the bottom port of both vertices is still unused.
  We draw the edge $(1,2)$ as a U-shape. Finally, we finish the
  layout by drawing the edge $(1,n)$. By construction, the left
  port of vertex~1 is still unused. Note that vertex~$n$ has no
  outgoing edges, so the top port of~$n$ is still free.  Hence,
  we can draw the edge $(1,n)$ as a horizontal or vertical segment
  followed by a three-quarter-circle. This completes the proof of
  Theorem~\ref{thm:bicon-sc2}.
\end{pf}


\section{Smooth Layouts for Arbitrary 4-Planar Graphs}
\label{sec:arbit}

In this section, we describe how to create \SC{2}-layouts for
arbitrary 4-planar graphs. To achieve this, we decompose the graph
into biconnected components, embed them separately and then connect
them.  For the connection it is important that one of the connector
vertices lies on the outer face of its component.  Within each
component, the connector vertices have degree at most~3; if they have
degree~2, we must make sure that their incident edges don't use
opposite ports.  Following Biedl and Kant~\cite{bk-bhogd-CGTA98}, we
say that a degree-2 vertex~$v$ is drawn \emph{with right angle} if the
edges incident to~$v$ use two neighboring ports.

\begin{lemma}
  \label{lem:rightangle}
  Any biconnected 4-planar graph admits an \SC2-layout such that all
  \mbox{degree-2} vertices are drawn with right angle.
\end{lemma}
\begin{pf}
  Let $v$ be a degree-2 vertex.  We now show how to adjust the
  algorithm of Section~\ref{sec:bicon} such that $v$ is drawn with
  right angle.  By construction, the top and the bottom ports of~$v$
  are used.  Let $(u,v)$ be the edge entering~$v$ from below (we allow
  $v=1$ and $u=2$).  We modify the algorithm such that $(u,v)$ uses
  the left or right rather than the bottom port of~$v$.  We consider
  three cases; $(u,v)$ is either L-shaped, U-shaped, or vertical.
  \textcolor{red}{These cases are handled when~$v$ is drawn inserted
  into the smooth orthogonal drawing.}

  First, we assume that $(u,v)$ is L-shaped; see
  Fig.~\ref{fig:rightangle-1}. Then, we can simply move $v$ to the
  same row as $u$, making the edge horizontal.

  Now, we assume that $(u,v)$ is U-shaped; see
  Fig.~\ref{fig:rightangle-2},~\ref{fig:rightangle-3}. Then $u=1$ and $v=2$ or
  vice versa.  If both have degree~2, we move the higher vertex to the
  row of the lower vertex (if necessary) and replace the U-shaped edge
  by a horizontal edge.  Otherwise we move the vertex with degree-2,
  say~$v$, downwards to row $y(u)-\Delta x$ such that we can replace
  the U-shape by an L-shape.

  Otherwise, $(u,v)$ is vertical; see Fig.~\ref{fig:rightangle-4}.
  Then, we compute a cut that starts slightly below~$v$,
  follows~$(u,v)$ downwards, passing~$u$ to its left. We move all
  vertices (including~$u$, but not~$v$) that lie on the right side of
  this cut (by at least $\Delta y$) to the right.  Then we can draw
  $(u,v)$ as an L-shape that uses the right port of~$v$.

  Observe that, in each of the three cases, we redraw all affected
  edges with~\SC2.  Hence, the modified algorithm still yields an
  \SC2-layout.  At the same time, all degree-2 vertices are drawn with
  right angle as desired.
\end{pf}

\begin{figure}[tb]
  \subfloat[\label{fig:rightangle-1}{an L-shape becomes a horizontal edge}]
  {\parbox[b]{.49\textwidth}{\centering\includegraphics[page=2]{rightangle}}}
  \hfill
  \subfloat[\label{fig:rightangle-2}{a U-shape becomes a horizontal edge}]%
  {\parbox[b]{.49\textwidth}{\centering\includegraphics[page=3]{rightangle}}}

  \subfloat[\label{fig:rightangle-3}{a U-shape becomes an L-shape}]%
  {\parbox[b]{.49\textwidth}{\centering\includegraphics[page=4]{rightangle}}}
  \hfill
  \subfloat[\label{fig:rightangle-4}{a vertical edge becomes an L-shape}]%
  {\parbox[b]{.49\textwidth}{\centering\includegraphics[page=1]{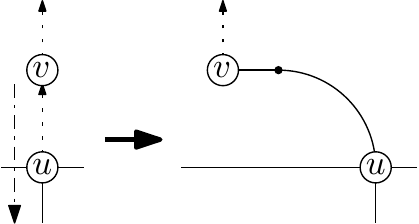}}}

  \caption{Modification of the placement of degree-2 vertices.}
  \label{fig:rightangle}
\end{figure}

Now we describe how to connect the biconnected components. Recall that
a \emph{bridge} is an edge whose removal disconnects a graph $G$.  We
call the two endpoints of a bridge \emph{bridge heads}.  A \emph{cut
  vertex} is a vertex whose removal disconnects the graph, but is not
a bridge head.

\begin{theorem}
  \label{thm:ar-smooth-2}
  Any 4-planar graph admits an \SC{2}-layout. 
\end{theorem}
\begin{proof}
  Let $G_0$ be some biconnected component of $G$, and let
  $v_1,\dots,v_k$ be the cut vertices and bridge heads of $G$ in $G_0$.
  For $i=1,\dots,k$, if $v_i$ is a bridge head, let $v_i'$ be the other
  head of the bridge, otherwise let $v_i'=v_i$.  Let~$G_i$ be the
  subgraph of $G$ containing $v_i'$ and the connected components
  of~$G-v_i'$ not containing~$G_0$. Following Lemma~\ref{lem:rightangle}, $G_0$ can be drawn such
  that all degree-2 vertices are drawn with right angles.

  The algorithm of Section~\ref{sec:bicon} that we modified in the proof
  of Lemma~\ref{lem:rightangle} places the last vertex ($n$) at the top
  of the drawing and thus on the outer face.  When drawing~$G_i$, we
  choose $v_i'$ as this vertex.  By induction, $G_i$ can be drawn such
  that all degree-2 vertices are drawn with right angles.

  In order to connect $G_i$ to $G_0$, we make~$G_0$ large enough to
  fit~$G_i$ into the face that contains the free ports of~$v_i$.  We may
  have to rotate~$G_i$ by a multiple of~$90^\circ$ to achieve the
  following.  If $v_i$ is a cut vertex, we make sure that $v_i'$ uses
  the ports of $v_i$ that are free in~$G_0$.  Then we identify $v_i$ and
  $v_i'$.  Otherwise we make sure that a free port of~$v_i$ and a free
  port of~$v_i'$ are opposite.  Then we draw the bridge $(v_i,v_i')$
  horizontally or vertically.  This completes our proof.
\end{proof}

For an example run of our algorithm, see Fig.~\ref{fig:example} in Appendix~C.
For graphs of maximum degree~3, we can make our drawings more compact.
This is due to the fact that we can avoid C-shaped edges (and hence
cuts) completely.  In the presence of L-shapes only, it suffices to
stretch the orthogonal drawing by a factor of~$n$.

\begin{theorem}
  \label{thm:bi-poly-3}
  Every biconnected 3-planar graph with $n$ vertices admits an
  \SC{2}-layout using area $\lfloor n^2/4 \rfloor \times \lfloor n/2
  \rfloor$.
\end{theorem}

\begin{wrapfigure}[6]{r}{.265\textwidth}
  \centering
  \vspace{-5ex} 
  \includegraphics{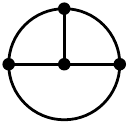}
  \caption{\SC1-layout of~$K_4$.}
  \label{fig:k4}
\end{wrapfigure}

\noindent\emph{Proof.}~
It is known that every biconnected 3-planar graph except $K_4$ has an
\OC2-layout using area $\lfloor n/2 \rfloor \times \lfloor n/2
\rfloor$ from Kant~\cite{k-dpwco-A96}. Now we use the same global
stretching as Bekos et al.\ when they showed that every \OC{2}-layout
can be transformed into an \SC{2}-layout \cite[Thm.~2]{BKKS12}: we
stretch the drawing horizontally by the height of the drawing, that
is, by a factor of $\lfloor n/2 \rfloor$.  This makes sure that we can
replace every bend by a quarter circle without introducing crossings.
Figure~\ref{fig:k4} shows an \SC1-layout of $K_4$; completing our
proof.  \hfill$\Box$

\section{\SC{1}-Layouts of Biconnected 4-Outerplane Graphs}
\label{sec:smooth-1}

In this section, we consider \emph{4-outerplane} graphs, that is,
4-outerplanar graphs with an outerplanar embedding. We prove that
any biconnected 4-outerplane graph admits an \SC1-layout. To do so,
we first prove the result for a subclass of 4-outerplane graphs,
which we call $(2,3)$-restricted outerplane graphs; then we
generalize. Recall that the \emph{weak dual} of a plane graph is the
subgraph of the dual graph whose vertices correspond to the bounded
faces of the primal graph.

\begin{definition}
  A 4-outerplane graph is called \emph{$(2,3)$-restricted} if it
  contains a pair of consecutive vertices on the outer face, $x$ and
  $y$, such that $\deg(x)=2$ and $\deg(y) \le 3$.
\end{definition}

\begin{lemma}
  \label{lem:2-3-restricted}
  Any biconnected $(2,3)$-restricted 4-outerplane graph admits an \SC{1}-layout.
\end{lemma}
\begin{pf}
  Let $x$ and $y$ be two vertices, consecutive on the outer face of
  the given graph~$G$ such that $\deg(x)=2$ and $\deg(y) \le 3$. Let also $T$ be the
  weak dual tree of $G$ rooted at the node, say $v^*$, of $T$
  corresponding to the bounded face, say $f^*$, containing both $x$ and $y$.
  We construct the \SC{1}-layout
  $\Gamma$ for $G$ by traversing~$T$, starting with $v^*$.  When
  we traverse a node of $T$, we draw the corresponding face of
  $G$ with \SC{1}.

  Consider the case when we have constructed a drawing $\Gamma(H)$ for
  a connected subgraph $H$ of $G$ and we want to add a new face~$f$
  to~$\Gamma(H)$.  For each vertex $u$ of $H$, let $p_u=(x(u), y(u))$
  denote the point at which $u$ is drawn in $\Gamma(H)$.  The
  \textit{remaining degree} of $u$ is the number of vertices adjacent
  to~$u$ in $G-H$.  Since we construct~$\Gamma(H)$ face by face, the
  remaining degree of each vertex in $H$ is at most two.  The
  \textit{free} ports of~$u$ are the ones that are not occupied by an
  edge of $H$ in $\Gamma(H)$.  During the construction of~$\Gamma$, we
  maintain the following four invariants:

  \begin{enumerate}[$(J_1)$]
  \item \label{ip1} $\Gamma(H)$ is an \SC{1}-layout that preserves the
    planar embedding of $G$, and each edge is drawn either as an axis-parallel line
    segment or as a quarter-circle in $\Gamma(H)$.  (Note that we do
    not use semi- and 3/4-circles.)

  \item \label{ip3} For each vertex $u$ of $H$, the free ports of $u$
    in $\Gamma(H)$ are consecutive around $u$, and they point to the
    outer face of $\Gamma(H)$.

  \item \label{ip4} Vertices with remaining degree exactly~2 are
    incident to an edge drawn as a quarter-circle.

  \item \label{ip5} If an edge $(u,v)$ is drawn as an axis-parallel
    segment, then at least one of $u$ and $v$ has remaining degree at
    most~1.  If $(u,v)$ is vertical and $y(u) < y(v)$, then $u$ has
    remaining degree at most~1 and the free port of $u$ in $\Gamma(H)$
    is horizontal; see Figs.~\ref{fig:outer-sc1}a, \ref{fig:outer-sc1}d
    and \ref{fig:outer-sc1}g.  Symmetrically,
    if $(u,v)$ is horizontal and $x(u) < x(v)$, then $u$ has remaining
    degree at most~1 and the free port of $u$ in $\Gamma(H)$ is
    vertical; see Figs.~\ref{fig:outer-sc1}b, \ref{fig:outer-sc1}e
    and \ref{fig:outer-sc1}h.
  \end{enumerate}

  \begin{figure}[tb]
    \centering
    \includegraphics[width=\textwidth]{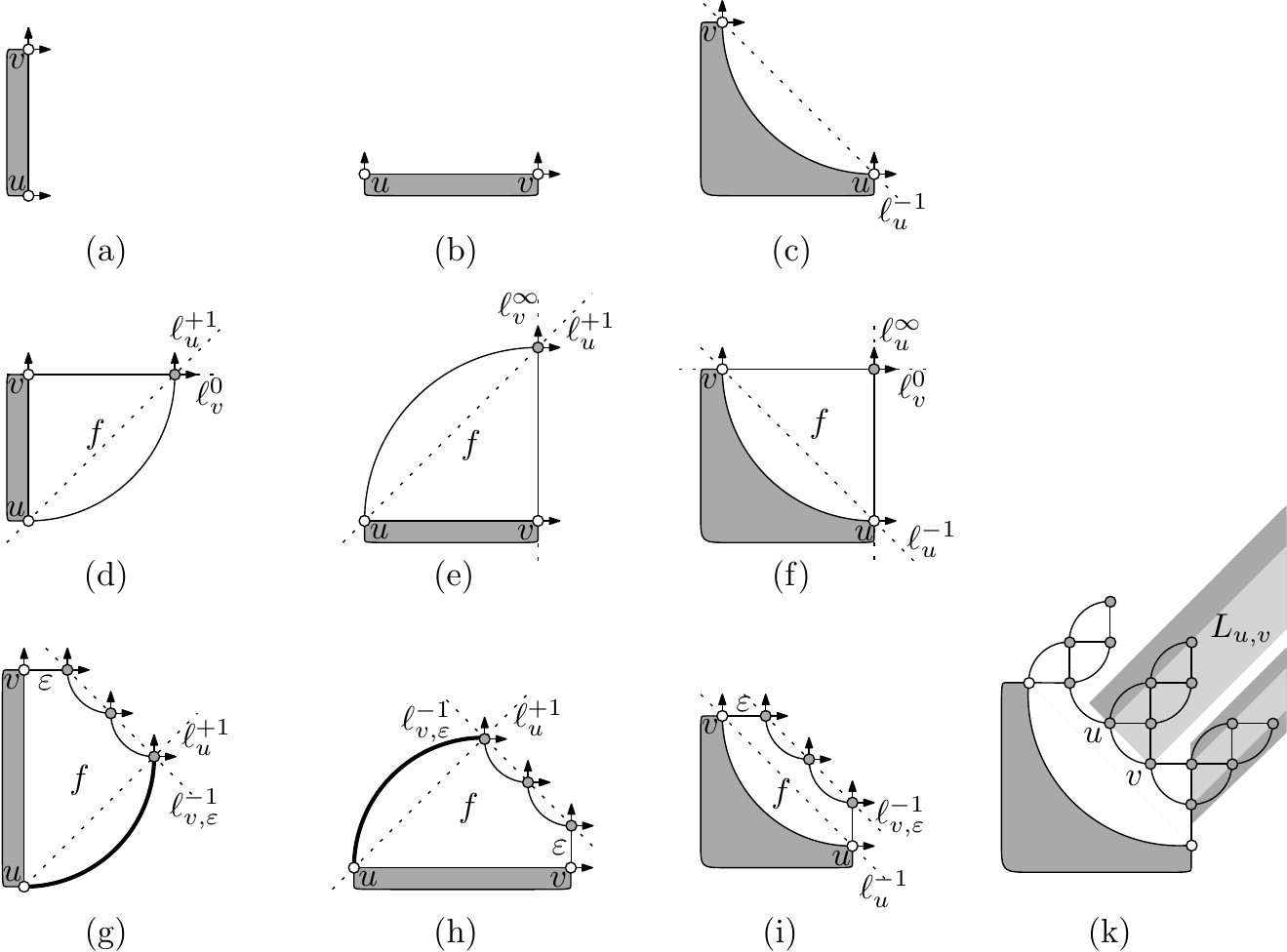}
    \caption{(a)-(i)~Different cases that arise when drawing face $f$ of $G$. (k)~A sample drawing.}
    \label{fig:outer-sc1}
  \end{figure}

  We now show how we add the drawing of the new face $f$ to
  $\Gamma(H)$.  Since~$G$ is biconnected and outerplanar, and due to
  the order in which we process the faces of~$G$, $f$ has exactly two
  vertices, say $u$ and $v$, which have already been drawn (as~$p_u$
  and~$p_v$).  The two
  vertices are adjacent.  Depending on how the edge $(u,v)$ is drawn
  in $\Gamma(H)$, we draw the remaining vertices and edges of $f$.

  Let $k \geq 3$ be the number of vertices on the boundary of~$f$.
  The slopes of the line segment~$\overline{p_up_v}$ is in
  $\{-1,0,+1,\infty\}$, where $\infty$ means that $\overline{p_up_v}$ is
  vertical.  For $s \in \{-1,0,+1,\infty\}$, we denote by $\ell_u^s$
  the line with slope~$s$ through~$p_u$.  Similarly, we denote by
  $\ell_{u,\eps}^s$ the line with slope~$s$ through the point
  $(x(u)+\eps,y(u))$, for some $\eps>0$. 
  Figs.~\ref{fig:outer-sc1}d--\ref{fig:outer-sc1}f show the
  drawing of $f$ for $k=3$, and
  Figs.~\ref{fig:outer-sc1}g--\ref{fig:outer-sc1}i for any $k \geq
  4$. 

  Note that the lengths of the line segments and the radii of the
  quarter-circles that form $f$ are equal (except for the radii of the
  bold-drawn quarter-circles of Figs.~\ref{fig:outer-sc1}g and
  \ref{fig:outer-sc1}h which are determined by the remaining edges of
  $f$). Hence, the lengths of the line segments and the radii of the
  quarter-circles that form any face that is descendant of face~$f$
  in~$T$ are smaller than or equal to the lengths of the line segments
  and the radii of the quarter-circles that form $f$. Our construction
  ensures that all vertices of the subgraph of $G$ induced by the
  subtree of $T$ rooted at $f$ lie in the interior or on the boundary
  of the diagonal semi-strip~$L_{uv}$ delimited by~$\ell_{u}^{+1}$,
  $\ell_{v}^{+1}$, and $\overline{p_up_v}$ (see Fig.~\ref{fig:outer-sc1}k).  The only edges of this
  subgraph that are drawn in the complement of~$L_{uv}$ (and are
  potentially involved in crossings) are incident to two vertices that
  both lie on the boundary of~$L_{uv}$.  In this particular case,
  however, the degree restriction implies that~$L_{uv}$ is surrounded
  from above and/or below by two empty diagonal semi-strips of at least
  half the width of semi-strip~$L_{uv}$, which is enough to ensure
  planarity for two reasons.
  First, any face that is descendant of face $f$ in $T$ is formed by line
  segments and quarter-circles of radius that are at most as big as the
  corresponding ones of face $f$.  Second, due to the degree
  restrictions, if two neighboring children of~$f$ are triangles, the
  left one cannot have a right child and vice versa.

  Let us summarize.  Fig.~\ref{fig:outer-sc1}d--\ref{fig:outer-sc1}i
  show that the drawing of $f$ ensures that
  invariants~($J_{\ref{ip1}}$)--($J_{\ref{ip5}}$) of our algorithm are
  satisfied for $H\cup\{f\}$.  We begin by drawing the root
  face~$f^*$. Since~$G$ is $(2,3)$-restricted, $f^*$ has two
  vertices~$x$ and~$y$ consecutive on the outer face with $\deg(x)=2$
  and $\deg(y) \leq 3$.  We draw edge $(x,y)$ as a vertical line
  segment.  Then the remaining degrees of~$x$ and~$y$ are 1 and 2,
  respectively, which satisfies the invariants for face~$f^*$.  Hence,
  we complete the drawing of~$f^*$ as in Fig.~\ref{fig:outer-sc1}d
  or~\ref{fig:outer-sc1}g.  Traversing~$T$ in pre-order, we complete
  the drawing of~$G$.
\end{pf}


Next, we show how to
deal with general biconnected 4-outerplane graphs.
Suppose $G$ is not $(2,3)$-restricted.  As the
following lemma asserts, we can always construct a biconnected $(2,3)$-restricted
4-outerplane graph by deleting a vertex of degree~2
from~$G$.

\begin{lemma}
  \label{lem:make-2-3-restricted}
  Let $G=(V,E)$ be a biconnected 4-outerplane graph that is not
  $(2,3)$-restricted.  Then $G$ has a degree-2 vertex whose
  removal yields a $(2,3)$-restricted biconnected 4-outerplane graph.
\end{lemma}
\begin{pf}
  The proof is by induction on the number of vertices. The base case
  is a maximal biconnected outerplane graph on six vertices, which is the only
  non-$(2,3)$-restricted graph with six or less vertices.  It is easy to
  see that in this case the removal of any degree-2 vertex yields a biconnected
  $(2,3)$-restricted 4-outerplane graph. Now assume that the
  hypothesis holds for any biconnected 4-outerplane graph with $k \geq 6$
  vertices. Let $G_{k+1}$ be a biconnected 4-outerplane graph on $k+1$ vertices,
  which is not $(2,3)$-restricted.  Let $\mathcal{F}$ be a face of
  $G_{k+1}$ that is a leaf in its weak dual. Then $\mathcal{F}$
  contains only one internal edge and exactly two external edges
  since, if it contained more than two external edges, $G_{k+1}$ would
  be $(2,3)$-restricted.  Therefore, $\mathcal{F}$ consists of
  three vertices, say $a$, $b$ and $c$, consecutive on the outer face
  and $\deg(a)=\deg(c)=4$, since otherwise $G_{k+1}$ would be
  $(2,3)$-restricted.  By removing $b$, we obtain a new graph, say
  $G_k$, on $k$ vertices.  If $a$ or $c$ is incident to a degree-2
  vertex in $G_k$, then $G_k$ is $(2,3)$-restricted.  Otherwise, by our
  induction hypothesis, $G_k$ has a degree-2 vertex whose removal
  yields a $(2,3)$-restricted outerplanar graph. Since this vertex is
  neither adjacent to $a$ nor $c$, the removal of this vertex
  makes $G_{k+1}$, too, $(2,3)$-restricted.
\end{pf}

\begin{theorem}
  \label{thm:biconnected-outerplanar}
  Any biconnected 4-outerplane graph admits an \SC1-layout.
\end{theorem}
\begin{pf}
  If the given graph~$G$ is $(2,3)$-restricted, then the result
  follows from Lemma~\ref{lem:2-3-restricted}.  Thus, assume that $G$
  is not $(2,3)$-restricted. Then, $G$ contains a degree-2 vertex, say
  $b$, whose removal yields a biconnected $(2,3)$-restricted 4-outerplane graph,
  say $G'$. Hence, we can apply the algorithm of
  Lemma~\ref{lem:2-3-restricted} to $G'$ and obtain an outerplanar
  \SC{1}-layout $\Gamma(G')$ of $G'$. Since this algorithm always
  maintains consecutive free ports for each vertex and the neighbors
  of $b$ are on the outer face of $\Gamma(G')$, we insert
  insert~$b$ and its two incident edges to obtain an \SC1-layout
  $\Gamma(G)$ of~$G$ as follows.  Let~$a$ and~$c$ be the neighbors
  of $b$ and assume w.l.o.g.\ that $c$ is drawn above $a$. If edge
  $(a,c)$ is drawn as a quarter-circle, then a 3/4-circle arc
  from~$p_c$ to $p_b$ and a quarter-circle from~$p_b$ to~$p_a$
  suffice.  Otherwise, line segment $\overline{p_ap_b}$ and a
  quarter-circle from~$p_b$ to~$p_c$ do the job.
\end{pf}

\section{Conclusions and Open Problems}
\label{sec:conclusions}

In this paper, we presented several new results about smooth
orthogonal drawings of 4-planar graphs.  Many problems remain open,
for example:
\begin{enumerate}
\item Can all 4-planar graphs be drawn in polynomial area with \SC2?
  We have shown this only for 3-planar graphs.
\item Identify larger classes of graphs admitting \SC1-layouts, e.g.,
  do all (not necessarily biconnected) 4-outerplanar or all 3-planar
  graphs admit \SC1-layouts?
\item We strongly conjecture that it is NP-hard to decide whether a
  4-planar graph has an \SC1-layout, but we struggled with some details
  in our attempt for a proof.
\end{enumerate}

\bibliographystyle{abbrv}
\bibliography{abbrv,lncs,stephen,maplab,smooth,lombardi,grid_drawing}

\newpage
\appendix



\section{A Lower Bound for the Area Requirement of \SC1-Layouts}
\label{sec:area}

In this section, we demonstrate an infinite family of 4-planar
graphs that require exponential area if they are drawn with \SC1.
Bekos et al.~\cite{BKKS12} presented such a family of graphs for the
rather restricted setting where both the embedding of the graph and
the port assignment of the edges are fixed. Here, we strengthen this
result. Consider the graph shown in Fig.~\ref{fig:exponarea-1}. This
graph consists of several layers.  Each layer consists of a cycle of
four pairs of adjacent triangles. The \SC1-layout of this graph in
Fig.~\ref{fig:exponarea-2} obviously requires exponential area since
every layer uses more than twice the area of the previous layer. We
will now show that this is the only \SC1-layout of the graph, up to
translation, rotation and scaling.

\begin{figure}[h]
  \hfill
  \subfloat[\label{fig:exponarea-1}]{\includegraphics[page=1]{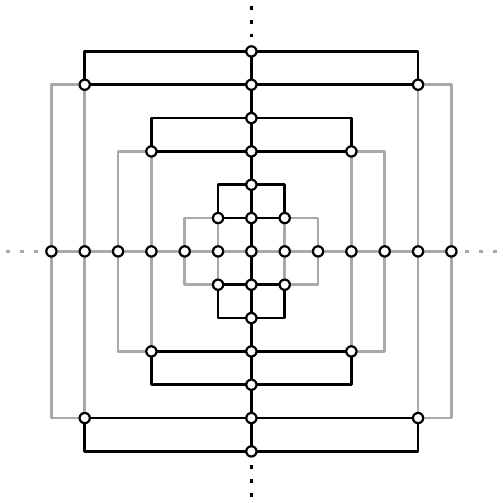}}
  \hfill
  \subfloat[\label{fig:exponarea-2}]{\includegraphics[page=2]{exponarea}}
  \hfill\null

  \caption{(a) A graph with an \OC{2}-layout using polynomial area (left)
    and (b) an \SC{1}-layout using exponential area (right).}
    \label{fig:exponarea}
\end{figure}

First, we show that there are only two ways to draw one of the
triangles of each layer. In Fig.~\ref{fig:alltriangles}, we show all
16 possible ways to get an \SC{1}-layout of a triangle. However, in
our graph all free ports have to lie on the outer face. There are
only two \SC1-layouts of a triangle that have this property, marked
by a
dashed circle. %

\begin{figure}[htbp]
  \hfill
  \begin{minipage}[b]{.23\textwidth}
     \centering
     {\includegraphics[page=1]{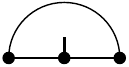}}
  \end{minipage}
  \hfill
  \begin{minipage}[b]{.23\textwidth}
     \centering
     {\includegraphics[page=2]{alltriangles}}
  \end{minipage}
  \hfill
  \begin{minipage}[b]{.23\textwidth}
     \centering
     {\includegraphics[page=3]{alltriangles}}
  \end{minipage}
  \hfill
  \begin{minipage}[b]{.23\textwidth}
     \centering
     {\includegraphics[page=4]{alltriangles}}
  \end{minipage}
  \hfill\null

  \smallskip

  \hfill
  \begin{minipage}[b]{.23\textwidth}
     \centering
     {\includegraphics[page=5]{alltriangles}}
  \end{minipage}
  \hfill
  \begin{minipage}[b]{.23\textwidth}
     \centering
     {\includegraphics[page=6]{alltriangles}}
  \end{minipage}
  \hfill
  \begin{minipage}[b]{.23\textwidth}
     \centering
     {\includegraphics[page=7]{alltriangles}}
  \end{minipage}
  \hfill
  \begin{minipage}[b]{.23\textwidth}
     \centering
     {\includegraphics[page=8]{alltriangles}}
  \end{minipage}
  \hfill\null

  \begin{minipage}[b]{.23\textwidth}
     \centering
     {\includegraphics[page=9]{alltriangles}}
  \end{minipage}
  \hfill
  \begin{minipage}[b]{.23\textwidth}
     \centering
     {\includegraphics[page=10]{alltriangles}}
  \end{minipage}
  \hfill
  \begin{minipage}[b]{.23\textwidth}
     \centering
     {\includegraphics[page=11]{alltriangles}}
  \end{minipage}
  \hfill
  \begin{minipage}[b]{.23\textwidth}
     \centering
     {\includegraphics[page=12]{alltriangles}}
  \end{minipage}
  \hfill\null

  \begin{minipage}[b]{.23\textwidth}
     \centering
     {\includegraphics[page=13]{alltriangles}}
  \end{minipage}
  \hfill
  \begin{minipage}[b]{.23\textwidth}
     \centering
     {\includegraphics[page=14]{alltriangles}}
  \end{minipage}
  \hfill
  \begin{minipage}[b]{.23\textwidth}
     \centering
     {\includegraphics[page=15]{alltriangles}}
  \end{minipage}
  \hfill
  \begin{minipage}[b]{.23\textwidth}
     \centering
     {\includegraphics[page=16]{alltriangles}}
  \end{minipage}
  \hfill\null

  \medskip

  \caption{All possible ways to get an \SC{1}-layout of a
    triangle. Only two of these drawings (enclosed by dashed red
    circles) have all their ports on the outer face.}
    \label{fig:alltriangles}

  \bigskip

   \hfill
  \includegraphics[page=2]{trianglecomb}
  \hfill
  \includegraphics[page=3]{trianglecomb}
  \hfill
  \includegraphics[page=4]{trianglecomb}
  \hfill\null

  \caption{There are three ways to draw two adjacent triangles.}
  \label{fig:trianglecomb}

  \bigskip

  \centering
  \begin{tabularx}{\textwidth}{@{}rXrXr@{}}
    \includegraphics[page=1]{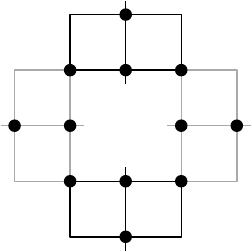} &&
    \includegraphics[page=2]{exponconstruct} &&
    \includegraphics[page=3]{exponconstruct} \\[-3.1ex]
    (a) && (b) && (c)
  \end{tabularx}

  \medskip

  \caption{(a) One layer of the graph in Fig.~\ref{fig:exponarea}. (b) \& (c)
    The only two ways to draw the subgraph depicted in (a) with \SC{1}.}
  \label{fig:exponconstruct}
\end{figure}

Next, we build a pair of adjacent triangles. In
Fig.~\ref{fig:trianglecomb}, we show that there are three ways to
combine two triangles that share an edge. Finally, we combine four
pairs of adjacent triangles to one layer of the graph. Using careful
case analysis, it can be shown that there are only two ways to draw
one of the layers with \SC1; see Fig.~\ref{fig:exponconstruct}.
However, it is easy to see that it is impossible to connect the
drawing shown in Fig.~\ref{fig:exponconstruct}c to another layer.
Thus, the \SC1-layout shown in Fig.~\ref{fig:exponarea-2} is the
only way to draw this graph, which proves the following theorem.

\begin{backInTime}{exponential-area}
\begin{theorem} 
  There is an infinite family of graphs that require exponential
  area if they are drawn with \SC1.
\end{theorem}
\end{backInTime}

\section{Biconnected Graphs without  \SC1-Layouts}
\label{sec:sc2}

In this section, we demonstrate an infinite family of biconnected
4-planar graphs that admit \OC2-layouts, but do not admit
\SC1-layouts. Bekos et al.~\cite{BKKS12} presented such a family of
graphs assuming a rather restricted setting in which the choice of
the outerface is fixed and always corresponds to a triangle.  Here,
we strengthen this results by providing an infinite family of
biconnected 4-planar graphs that admit no \SC1-layout in any
embedding. We start with the following lemma.

\begin{lemma}\label{lem:sc2}
 There exists a biconnected 4-planar graph that admits
  an \OC2-layout, but does not admit an \SC{1}-layout.
\end{lemma}
\begin{pf} Let $G$ be the graph of Fig.~\ref{fig:oc2-sc2-1}. We prove that $G$ has no \SC{1}-layout.
 First, note that $G$ contains two copies of the graph depicted in Fig.~\ref{fig:oc2-sc2-2}. We denote
 this graph by $H$. We first prove that $H$ has no \SC{1}-layout with the given embedding.
 In particular, we show that the subgraph of $H$ induced by the vertices on or inside the black cycle
 cannot be drawn with \SC1.

\begin{figure}[ht]
  \centering
  \subfloat[\label{fig:oc2-sc2-1}]
  {\includegraphics[page=1,scale=1.2]{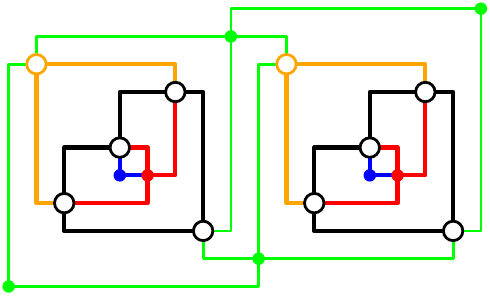}}
  \hfil
  \subfloat[\label{fig:oc2-sc2-2}]
  {\includegraphics[page=1,scale=1.2]{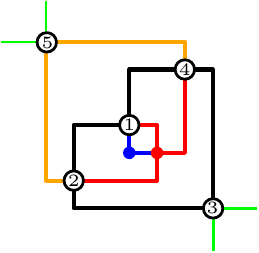}}

  \caption{(a) A graph that admits an \OC2-layout, but does
    not admit an \SC1-layout; (b) the important part of~(a) in detail.}
  \label{fig:oc2-sc2}
\end{figure}

\begin{figure}[ht]
  \begin{tabular}{cccc}
    \includegraphics[page=2,scale=1.2]{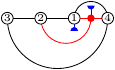} &
    \includegraphics[page=3,scale=1.2]{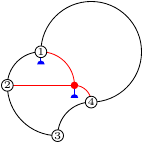} &
    \includegraphics[page=4,scale=1.2]{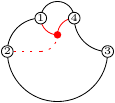} &
    \includegraphics[page=5,scale=1.2]{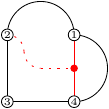} \\
    (a) & (b) & (c) & (d) \\
    \includegraphics[page=6,scale=1.2]{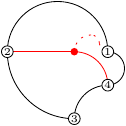} &
    \includegraphics[page=7,scale=1.2]{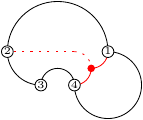} &
    \includegraphics[page=8,scale=1.2]{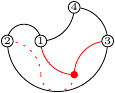} &
    \includegraphics[page=9,scale=1.2]{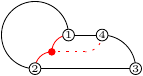} \\
    (d) & (e) & (f) & (g) \\
    \includegraphics[page=10,scale=1.2]{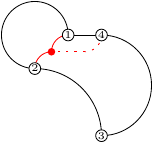} &
    \includegraphics[page=11,scale=1.2]{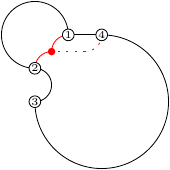} &
    \includegraphics[page=12,scale=1.2]{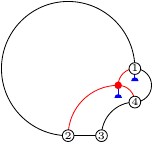} &
    \includegraphics[page=13,scale=1.2]{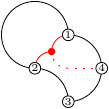} \\
    (h) & (i) & (j) & (k) \\
  \end{tabular}
  \caption{Illustration for the proof of Lemma~\ref{lem:sc2}}
  \label{fig:sc2}
\end{figure}

Consider edge $e=(1,2)$ of~$H$. This edge can be drawn as a
straight-line segment, quarter circle, semi-circle or 3/4-circle.
Fig.~\ref{fig:sc2}a illustrates the case where $e$ is drawn as a
horizontal line segment.  In this case, the ports for the edges are
fixed due to the given embedding and it is not possible to complete
the drawing. The case where $e$ is drawn as a vertical segment is
analogous. Similarly, we show that there is no \SC1-layout for $H$
if $e$ is drawn as a quarter-circle in
Figs.~\ref{fig:sc2}b--\ref{fig:sc2}c, as a semi-circle in
Figs.~\ref{fig:sc2}d--\ref{fig:sc2}g and as a 3/4-circle in
Figs.~\ref{fig:sc2}h--\ref{fig:sc2}l.  Thus, there is no
\SC{1}-layout for this fixed embedding of~$H$.

Next, we claim that there is no \SC1-layout for any embedding of $H$
where the vertices~2, 3, 4, and~5 define the outer cycle. Indeed, if
the outerface is fixed, then the only way to find a different
embedding is to find a separating pair $\{u,v\}$ in $H$ and ``flip''
one of the components of $H-\{u,v\}$.  There are two possible
separating pairs in~$H$: (i)~vertex~1 and the red vertex; then the
flip with respect to this pair gives an isomorphic graph due to
symmetry; and (ii)~vertices~2 and~4; then the flip with respect to
this pair again gives an isomorphic graph by interchanging the role
of~3 and~5. Thus, with the fixed outer cycle $(2,3,4,5)$, all
possible embeddings of~$H$ are isomorphic. Since~$G$ contains two
copies of~$H$, in any embedding of~$G$, at least one of the copies
will retain its outer cycle. Hence, there is no \SC{1}-drawing for
any embedding of~$G$.
\end{pf}

Graph $G$ of Fig.~\ref{fig:oc2-sc2-1} uses a few short paths to
connect two copies of $H$.  Obviously, we can add an arbitrary
number of vertices to these paths such that the augmented graph
remains biconnected and 4-planar.  This proves the following
theorem.

\begin{backInTime}{OCtwo-not-SCone}
\begin{theorem}
  There is an infinite family of biconnected 4-planar graphs that
  admit \OC2-layouts but do not admit \SC{1}-layouts.
\end{theorem}
\end{backInTime}

%
%
%

\clearpage

\begin{figure}[h!]
  \textbf{\large C~~{An Example Run of Our Algorithm for \SC2-Layout}}
  \bigskip


\begin{minipage}[b]{.48\textwidth}
  \centering
  \includegraphics[scale=.63,page=1]{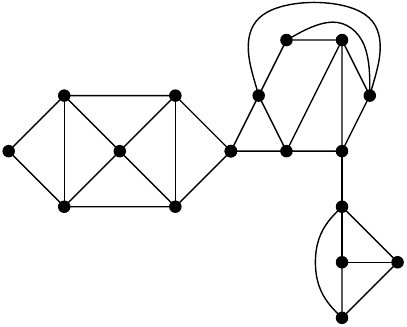}
  \caption*{(a) input graph}
\end{minipage}
\hfill
\begin{minipage}[b]{.48\textwidth}
  \centering
  \includegraphics[scale=.63,page=2]{pic/example_nonbicon_new}
  \caption*{(b) biconnected subgraphs}
\end{minipage}

\medskip

\begin{minipage}[b]{.48\textwidth}
  \centering
  \includegraphics[scale=.63,page=3]{pic/example_nonbicon_new}
  \caption*{(c) $st$-ordering}
\end{minipage}
\hfill
\begin{minipage}[b]{.48\textwidth}
  \centering
  \includegraphics[scale=.63,page=5]{pic/example_nonbicon_new}
  \caption*{(d) eliminating S-shapes}
\end{minipage}

\medskip

\begin{minipage}[b]{.48\textwidth}
  \centering
  \includegraphics[scale=.63,page=6]{pic/example_nonbicon_new}
  \caption*{(e) drawing degree-2 vertices with right angle}
\end{minipage}
\hfill
\begin{minipage}[b]{.48\textwidth}
  \centering
  \includegraphics[scale=.63,page=7]{pic/example_nonbicon_new}
  \caption*{(f) drawing subgraphs with \SC2}
\end{minipage}

\medskip

\begin{minipage}[b]{.48\textwidth}
  \centering
  \includegraphics[scale=.63,page=8]{pic/example_nonbicon_new}
  \caption*{(g) connecting second and fourth part by using the bridge (third part)}
\end{minipage}
\hfill
\begin{minipage}[b]{.48\textwidth}
  \centering
  \includegraphics[scale=.63,page=9]{pic/example_nonbicon_new}
  \caption*{(h) connecting first and second part}
\end{minipage}

\caption{An example-run of our Algorithm for \SC2-layout. The circle
  vertices of component~$i$ correspond to the cut vertex~$v_i'$.
  The square vertices correspond to cut vertices of other components.}
\label{fig:example}

\end{figure}

\paragraph{Acknowledgments.}

We thank Therese Biedl for pointers to Figure~\ref{fig:transport_map}
and the work of Liu et al.~\cite{lms-la2be-DAM98}.

\end{document}